\documentclass[noendauthors]{birkjour}
\usepackage[T1]{fontenc}
\usepackage{amssymb}
\usepackage{microtype}
\usepackage{etoolbox}
\usepackage{placeins}
\usepackage{float}
\usepackage[hidelinks]{hyperref}
\numberwithin{equation}{section}
\makeatletter
\patchcmd{\@setauthors}{\endtrivlist}{
  \@setaddresses\global\let\addresses\@empty
  \endtrivlist
}
\makeatother

\newcommand{\A}{\mathcal A}
\newcommand{\Hh}{\mathcal H}
\newcommand{\Tr}{\operatorname{Tr}}
\newcommand{\diag}{\operatorname{diag}}
\newcommand{\spec}{\operatorname{spec}}
\newcommand{\dist}{\mathrm d}
\newcommand{\norm}[1]{\lVert #1\rVert}
\newtheorem{prop}{Proposition}[section]

\begin{document}
\title{Helstrom geometry with time-dependent\\ spectral calibration}
\author{Kaushlendra Kumar}
\address{School of Mathematical Sciences\newline
Queen Mary University of London\newline
Mile End Road, London E1 4NS, United Kingdom}
\email{kaushlendra.kumar@qmul.ac.uk}
\keywords{Connes spectral distance, finite spectral geometry, information backflow, qubit trace distance}

\begin{abstract}
In a recent work we introduced a finite scalar--qubit--scalar spectral triple whose Connes distance reproduces the Bloch-ball trace metric up to a factor $2/\Lambda$, where $\Lambda^2=\mu_L^2+\mu_R^2$ for isotropic link strengths $\mu_L,\mu_R$ appearing in the proposed Dirac operator. Here we analyse time-dependent variation of these links, independent of the qubit dynamics. In doing so, we separate the distinguishability and calibration contributions to the spectral rate. A depolarizing example demonstrates spectral growth despite a monotonic contraction of the trace-distance. More generally, we characterize the admissible trace-distance histories compatible with a given positive spectral record and exhibit contracting and reviving realizations. At fixed $\Lambda$, we also study a sinusoidal link imbalance that preserves the Dirac spectrum and qubit metric, while the scalar--scalar distance oscillates with half the period of the imbalance.
\end{abstract}
\maketitle

\raggedbottom

\section{Introduction}\label{sec:intro}

For a quantum system, trace distance measures the distinguishability of states. Given a system prepared in one of two known states with density matrices $\rho$ and $\sigma$, each with equal prior probability, an optimal measurement identifies the preparation with success probability
\begin{equation}\label{eq:helstrom}
 p_{\rm succ}^{\rm opt}(\rho,\sigma)=\frac12\bigl(1+T(\rho,\sigma)\bigr),
 \qquad T(\rho,\sigma)=\frac12\norm{\rho-\sigma}_1.
\end{equation}
This is Helstrom's formula~\cite{Helstrom}, where $\norm{X}_1=\Tr\sqrt{X^*X}$ is the trace norm. We assume equal priors throughout.

On the other hand, there exists a separate notion of spectral distance in noncommutative geometry (NCG) à la Connes. Here, distance between states can be obtained once a spectral triple $(\A,\Hh,D)$ is specified for an involutive algebra $\A$ represented via $\pi$ on a Hilbert space $\Hh$ and $D$ a self-adjoint Dirac operator. For two states, i.e.\ normalized positive linear functionals $\phi,\psi:\A\to\mathbb C$, the Connes spectral distance reads~\cite{Connes}
\begin{equation}\label{eq:connes}
 \dist_D(\phi,\psi)=
 \sup_{\substack{a=a^*\\L_D(a)\leq1}}|\phi(a)-\psi(a)|,
 \qquad L_D(a)=\norm{[D,\pi(a)]}_{\rm op}.
\end{equation}
Here $L_D$ is the Lipschitz seminorm, defined using the operator norm. It should be noted that algebra elements $a\in\A$ with $L_D(a)=0$ can be scaled arbitrarily within the constraint, so the distance is infinite whenever $\phi-\psi$ does not vanish on them. All our data are finite dimensional, and no real structure (relevant in the reformulation of particle physics within NCG) is imposed.

For the qubit algebra $M_2(\mathbb C)$ acting irreducibly on $\mathbb C^2$, it turns out the commutator has such an unconstrained direction. A $2\times2$ Dirac matrix with distinct eigenvalues commutes with the Pauli matrix along its eigenaxis, so that only states with the same Bloch component along this axis are at finite distance~\cite{IKM}. The spin-$1/2$ fuzzy sphere removes the obstruction with an equivariant Dirac operator~\cite{DLV}. In our Hilbert--Schmidt operatorial treatment of the Moyal plane and fuzzy sphere~\cite{MoyalFuzzy}, we also recovered the pure-state chordal distance at $n=1/2$ using Dirac eigenspinors. In~\cite{Static} we removed the obstruction by attaching two scalar summands to the qubit block through identity links, which constrain all Pauli directions at once. The spectral distance between qubit states is then the trace distance up to a constant fixed by the two link strengths.

In this article we let the two links depend on time, independently of the quantum channel under which the qubit states evolve. For open quantum systems, a temporary increase of the trace distance between two states under a common dynamical map witnesses information backflow in the sense of Breuer, Laine and Piilo (BLP)~\cite{BLP,Review}. We ask whether spectral growth still has this meaning when the Dirac operator changes in time. We show that a raw spectral record does not determine the trace-distance history, while scalar distances at the same times fix the missing scale. Since only one combination of the two links enters the qubit metric, we also consider link deformations that preserve this combination and examine their effect on the scalar distances.

\section{Scalar-anchored Helstrom geometry}\label{sec:static}

We recall the finite geometry of~\cite{Static}, where these results are proved. Each scalar sector is represented on a copy of $\mathbb C^2$, allowing identity blocks as the off-diagonal Dirac links. The finite spectral data read
\begin{equation}\label{eq:data}
 \begin{gathered}
 \A=\mathbb C_L\oplus M_2(\mathbb C)\oplus\mathbb C_R,
 \quad \Hh=\mathbb C^2\oplus\mathbb C^2\oplus\mathbb C^2,\\
 \pi(x,m,z)=\diag(xI_2,m,zI_2),\\
 D=\begin{pmatrix}
 0&\mu_L I_2&0\\ \mu_L I_2&0&\mu_R I_2\\0&\mu_R I_2&0
 \end{pmatrix},\quad \mu_L,\mu_R>0.
 \end{gathered}
\end{equation}
The links compare a qubit observable with a scalar value in an isotropic manner without singling out a Bloch direction. For any qubit density matrix $\rho$, we identify the sector states as
\begin{equation}\label{eq:states}
 \omega_L(x,m,z)=x,\qquad
 \omega_\rho(x,m,z)=\Tr(\rho m),\qquad
 \omega_R(x,m,z)=z.
\end{equation}
Next, using $D=D_0\otimes I_2$, one finds
\begin{equation}\label{eq:spectrum}
 \begin{gathered}
 D_0=\begin{pmatrix}0&\mu_L&0\\\mu_L&0&\mu_R\\0&\mu_R&0\end{pmatrix},\\
 \det(uI_3-D_0)=u(u^2-\Lambda^2),
 \quad \Lambda=\sqrt{\mu_L^2+\mu_R^2},
 \end{gathered}
\end{equation}
so that $\spec(D)=\{-\Lambda,0,\Lambda\}$, where each eigenvalue is doubly degenerate.

\begin{figure}[tb]
 \centering
 \includegraphics[width=0.9\linewidth]{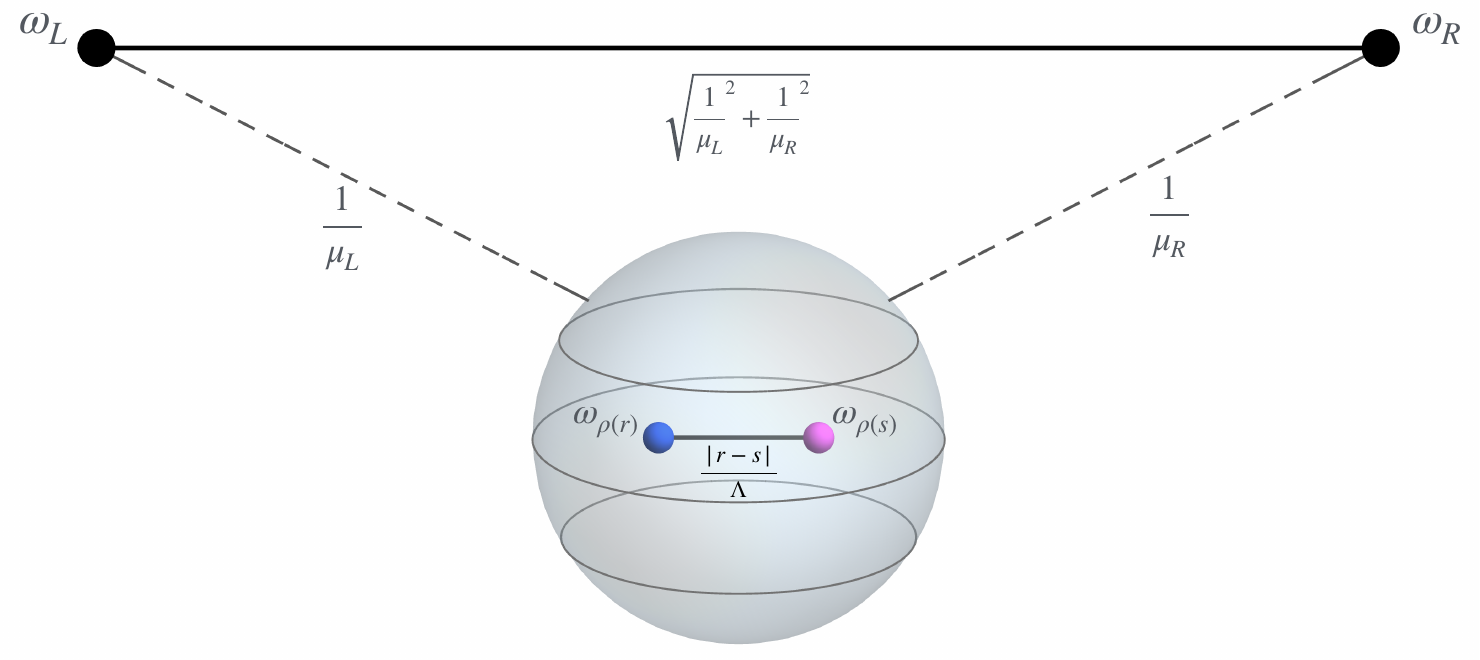}
 \caption{Illustration of scalar-anchored calibration geometry, adapted from~\cite{Static}. The scalar states $\omega_L$ and $\omega_R$ lie at distances $\ell_L=\mu_L^{-1}$ and $\ell_R=\mu_R^{-1}$ from every qubit state, with $\dist_D(\omega_L,\omega_R)^2=\ell_L^2+\ell_R^2$. Inside the qubit block the spectral distance is the calibrated chord $\dist_D(\omega_{\rho(r)},\omega_{\rho(s)})=\Lambda^{-1}|r-s|$.}
 \label{fig:static-geometry}
\end{figure}

For a self-adjoint element $(x,m,z)$, we diagonalize $m$ and use the same unitary in all three copies of $\mathbb C^2$, leaving the identity links unchanged. Reordering the basis then splits the commutator into two blocks, one for each eigenvalue $\lambda_\pm$ of $m$, and gives the Lipschitz seminorm~\cite{Static}
\begin{equation}\label{eq:channels}
 L_D(x,m,z)^2=
 \max_{i=\pm}\bigl\{\mu_L^2(\lambda_i-x)^2+
                            \mu_R^2(\lambda_i-z)^2\bigr\}.
\end{equation}
The middle-state difference is independent of the scalar entries. We can therefore minimize~\eqref{eq:channels} over those entries at fixed middle observable to determine whether it has an admissible scalar extension. Both scalars take the eigenvalue midpoint at the minimum. Optimizing the remaining eigenvalue gap and projector direction gives~\cite{Static}, for any two qubit density matrices,
\begin{equation}\label{eq:static}
 \dist_D(\omega_\rho,\omega_\sigma)
 =\frac{\norm{\rho-\sigma}_1}{\Lambda}
 =\frac{2}{\Lambda}T(\rho,\sigma).
\end{equation}
Using Bloch coordinates $\rho(r)=(I_2+r\cdot\sigma)/2$, with $\sigma$ denoting the Pauli vector and $r$ inside the Bloch ball $B^3:=\{r\in\mathbb R^3:|r|\leq1\}$, we obtain a chord length on the full mixed-state ball,
\begin{equation}\label{eq:chord}
 \dist_D(\omega_{\rho(r)},\omega_{\rho(s)})=\frac{|r-s|}{\Lambda}.
\end{equation}
The factor $\Lambda/2$ thus converts spectral distances between qubit states into trace distances.

Moreover, the scalar states lie at fixed distances from every qubit state,
\begin{equation}\label{eq:radii}
 \ell_L:=\dist_D(\omega_L,\omega_\rho)=\mu_L^{-1},
 \qquad
 \ell_R:=\dist_D(\omega_R,\omega_\rho)=\mu_R^{-1},
\end{equation}
while the distance between the two scalar states satisfies the Pythagorean relation
\begin{equation}\label{eq:pythagoras}
 \dist_D(\omega_L,\omega_R)
 =\sqrt{\mu_L^{-2}+\mu_R^{-2}}
 =\sqrt{\ell_L^2+\ell_R^2}.
\end{equation}
The scalar states in Fig.~\ref{fig:static-geometry} sit separately from the qubit states (shown as a colored ball). Since~\eqref{eq:radii} holds for every qubit state, the endpoints of the dashed links within the ball are arbitrary. A route through any such state has length $\ell_L+\ell_R$, strictly greater than the direct scalar distance~\eqref{eq:pythagoras}.

\section{Time-dependent spectral calibration}\label{sec:time}

We now let the couplings in~\eqref{eq:data} be positive differentiable functions of time, with $\A$, $\Hh$ and $\pi$ fixed. The qubit states evolve under a differentiable family $\{\Phi_t\}_{t\geq0}$ of completely positive trace-preserving (CPTP) maps\footnote{Complete positivity requires $(\Phi_t\otimes\mathrm{id}_n)(X)\geq0$ for every joint positive operator $X\geq0$ and every auxiliary dimension $n$. Trace preservation requires $\Tr[\Phi_t(Y)]=\Tr Y$ for every qubit operator $Y$.}, $\rho_t=\Phi_t(\rho_0)$ and $\sigma_t=\Phi_t(\sigma_0)$ with $\Phi_0=\mathrm{id}$. At this stage, we do not require the semigroup property, $\Phi_{t+s}=\Phi_t\circ\Phi_s$, for such CPTP maps. Also, the family is specified independently of $D(t)$, which only fixes the metric at time $t$. As long as the links stay proportional to $I_2$, the identity~\eqref{eq:static} holds at every instant,
\begin{equation}\label{eq:instant}
 \begin{gathered}
 \dist(t):=\dist_{D(t)}(\omega_{\rho_t},\omega_{\sigma_t})
 =\frac{2T(t)}{\Lambda(t)},\\
 T(t):=T(\rho_t,\sigma_t),
 \qquad \Lambda(t)=\sqrt{\mu_L(t)^2+\mu_R(t)^2}.
 \end{gathered}
\end{equation}
We can now use~\eqref{eq:instant} to compute the two rate contributions by writing $T=\Lambda\dist/2$ and differentiating the expression,
\begin{equation}\label{eq:rate}
 \dot T(t)=\frac12\bigl[\Lambda(t)\dot{\dist}(t)
                              +\dot\Lambda(t)\dist(t)\bigr].
\end{equation}
For $T(t)>0$ this relation can also be written in terms of fractional rates,
\begin{equation}\label{eq:lograte}
 \frac{\dot{\dist}(t)}{\dist(t)}
 =\frac{\dot T(t)}{T(t)}-\frac{\dot\Lambda(t)}{\Lambda(t)}.
\end{equation}
Under the assumed differentiability of the states, this logarithmic relation applies wherever $T(t)>0$, equivalently $\rho_t\neq\sigma_t$.

To relate these rates to the qubit dynamics, let us consider Bloch vectors $r(t)$ and $s(t)$ for the evolving states $\rho_t$ and $\sigma_t$, respectively, obeying a common time-local affine equation,
\begin{equation}\label{eq:affine}
 \dot r=A(t)r+c(t),\qquad \dot s=A(t)s+c(t),
 \qquad A(t)\in M_3(\mathbb R),\quad c(t)\in\mathbb R^3.
\end{equation}
We can now subtract the two equations to remove the common translation, such that their difference $q:=r-s$ obeys $\dot q=Aq$. This allows us to write the trace distance as $T=|q|/2$ using  \eqref{eq:static} and \eqref{eq:chord}, hence $q^{\mathsf T}q=4T^2$. Next, we can differentiate this identity following $\dot q=Aq$ and find, for $q\neq0$,
\begin{equation}\label{eq:affine-rate}
 \frac{\dot T}{T}
 =\frac{q^{\mathsf T}Sq}{q^{\mathsf T}q},
 \qquad S:=\frac12(A+A^{\mathsf T}),
\end{equation}
where the quadratic form depends only on the symmetric part $S$. The antisymmetric part produces an instantaneous rotation of $q$ and hence drops out of its length derivative. Substituting~\eqref{eq:affine-rate} into~\eqref{eq:lograte}, we then obtain the corresponding spectral rate,
\begin{equation}\label{eq:affine-spectral-rate}
 \frac{\dot{\dist}}{\dist}
 =\frac{q^{\mathsf T}Sq}{q^{\mathsf T}q}-\frac{\dot\Lambda}{\Lambda}.
\end{equation}

Let us now consider the depolarizing semigroup,
\begin{equation}\label{eq:depolarizing}
 \Phi_t(\rho)=e^{-\gamma t}\rho+(1-e^{-\gamma t})\frac{I_2}{2},
 \qquad \gamma>0,
\end{equation}
for which $\Phi_{t+s}=\Phi_t\circ\Phi_s$ and $\dot\rho_t=-\gamma(\rho_t-I_2/2)$ (see~\cite{GKS}). The corresponding Bloch-affine form \eqref{eq:affine} fixes $A=S=-\gamma I_3$ and $c=0$, and therefore integrating~\eqref{eq:affine-rate} gives
\begin{equation}\label{eq:depolarizing-trace}
 T(t)=T_0e^{-\gamma t},
\end{equation}
where $T_0>0$ is the initial trace distance.

To obtain spectral growth despite this contraction, we need a decreasing Dirac scale, and we consider the following simple choice that stays bounded away from zero,
\begin{equation}\label{eq:example-scale}
 \Lambda(t)=\Lambda_0\bigl(1-\eta+\eta e^{-\kappa t}\bigr),
 \qquad 0<\eta<1,\quad \kappa>0,
\end{equation}
with $\Lambda_0>0$, that can be realized for instance by equal links $\mu_L(t)=\mu_R(t)=\Lambda(t)/\sqrt2$. It decreases monotonically from $\Lambda_0$ to $\Lambda_0(1-\eta)>0$ and is independent of the qubit dynamics. Combining this scale with~\eqref{eq:depolarizing-trace} and substituting into~\eqref{eq:instant} and~\eqref{eq:lograte} gives
\begin{equation}\label{eq:example-distance}
 \begin{aligned}
 \dist(t)&=\frac{2T_0e^{-\gamma t}}
 {\Lambda_0(1-\eta+\eta e^{-\kappa t})},\\
 \frac{\dot{\dist}(t)}{\dist(t)}
 &=-\gamma+\frac{\eta\kappa e^{-\kappa t}}
 {1-\eta+\eta e^{-\kappa t}}.
 \end{aligned}
\end{equation}
The positive scale term decreases from $\eta\kappa$ to zero. If $\eta\kappa>\gamma$, it initially exceeds the contraction rate, causing the spectral length $\dist$ to grow until the two contributions cancel out, after which $d$ starts decreasing. Solving $\dot{\dist}(t_*)=0$ in~\eqref{eq:example-distance} gives
\begin{equation}\label{eq:turning}
 t_* =\frac1\kappa\log\frac{\eta(\kappa-\gamma)}{\gamma(1-\eta)}.
\end{equation}
Since $\eta\kappa>\gamma$ implies $\eta(\kappa-\gamma)>\gamma(1-\eta)$, the logarithm in~\eqref{eq:turning} has an argument greater than one and $t_*>0$.

For illustration, we choose $\gamma=0.20$, $\kappa=0.50$ and $\eta=0.56$ in fixed units of time, giving $t_*\simeq1.293$ from~\eqref{eq:turning}. The corresponding evolution is shown in Fig.~\ref{fig:two-effects}(a), with both distances normalized to their initial values. As time grows, the spectral distance attains its maximum of approximately $1.053\,\dist(0)$ at $t_*$, while the trace distance has decreased to approximately $0.772\,T_0$. This growth calls for a natural comparison with the BLP criterion~\cite{BLP,Review}; the latter identifies information backflow through an increase in trace distance. To this end, notice here that equation \eqref{eq:rate} gives $\dot T=\frac{1}{2}(\Lambda\dot{\dist}+\dot\Lambda\dist)=-\gamma T<0$ for every distinct evolving qubit pair, including in the interval $0<t<t_*$ where $\dot{\dist}>0$, meaning that the depolarizing evolution has {\em no} BLP backflow.

We saw in the depolarizing example above that $\Lambda(t)$ was given and $\dist(t)$ followed from it. Another interesting scenario is when only the raw record $\dist(t)$ of an evolving qubit pair is known. Which trace-distance histories can furnish this record if the scale is unknown? For any proposed positive history $T(t)$, equation~\eqref{eq:instant} fixes $\Lambda(t)=2T(t)/\dist(t)$, but such a history must also be realized by a common CPTP family. It is known \cite{Review} that contractivity requires $T(t)\leq T(0)\leq1$, and we show in the following how this condition is sufficient as well.

\begin{prop}[Compatible trace-distance histories]\label{prop:nonid}
Consider a given differentiable spectral-distance record $\dist(t)>0$ on $[0,t_1]$, along with a differentiable candidate trace-distance history $T(t)$. There exist two initial qubit states, a differentiable CPTP family $\Phi_t$ with $\Phi_0=\mathrm{id}$, and differentiable Dirac link strengths $\mu_L(t),\mu_R(t)>0$ realizing the trace-distance history $T(t)$ and spectral-distance record $\dist(t)$ through~\eqref{eq:instant} if and only if
\begin{equation}\label{eq:admissible}
 0<T(t)\leq T(0)\leq1,\qquad t\in[0,t_1].
\end{equation}
For such a realization, the calibration scale is $\Lambda(t)=2T(t)/\dist(t)$.
\end{prop}

\begin{proof}
In the forward direction, suppose such a realization exists, in which case the corresponding positive links naturally provide the scale $\Lambda(t)>0$. Owing to this, Eq.~\eqref{eq:instant} gives $T(t)=\Lambda(t)\dist(t)/2>0$, and contractivity under the common CPTP maps~\cite[Sec.~II.C.1]{Review} gives $T(t)\leq T(0)\leq1$.

Conversely, assuming a given history $T(t)$ obeying Eq. \eqref{eq:admissible}, we construct an explicit qubit evolution as follows. First, we take the initial Bloch vectors as $r_0=-s_0=(T(0),0,0)$, both of which lie in the Bloch ball since $T(0)\leq1$. Next, using $f(t):=T(t)/T(0)\in(0,1]$, we can employ the standard dephasing channel (a special case of the random-unitary qubit maps in~\cite[Sec.~III.A.3]{Review}),
\begin{equation}\label{eq:dephasing}
 \Phi_t(\rho)=\frac{1+f(t)}{2}\,\rho+\frac{1-f(t)}{2}\,\sigma_3\rho \sigma_3,
\end{equation}
with $\sigma_3$ the third Pauli matrix. The two non-negative coefficients add to one with $f(0)=1$ giving $\Phi_0=\mathrm{id}$, thereby producing a CPTP $\Phi_t$ as a convex combination of unitary conjugations. Now the conjugation by $\sigma_3$ reverses the first two Bloch components, so that by Eq. \eqref{eq:dephasing} the resulting Bloch vector has its first two components multiplied by $f(t)$ while the $z$-component remains invariant. This means that our chosen qubit pair evolves to produce $q(t)=(2f(t)T(0),0,0)$, and hence the required trace distance comes out to be $|q(t)|/2=f(t)T(0)=T(t)$.

All that remains is to choose suitable link strengths that reproduce the given spectral record. We can achieve this via the convenient choice $\mu_L(t)=\mu_R(t)=\Lambda(t)/\sqrt2$, with the positive differentiable scale $\Lambda(t)=2T(t)/\dist(t)$, which recovers the record $\dist(t)$ through Eq. \eqref{eq:instant}.
\end{proof}

In this way, we see that Proposition~\ref{prop:nonid} allows for all possible differentiable histories (not just those with monotonic decrease) that satisfy Eq. \eqref{eq:admissible}. For the depolarizing example~\eqref{eq:example-distance}, the corresponding bound becomes
\[\Lambda(t)\dist(t)\leq\Lambda(0)\dist(0)\leq2\] 
which excludes constant calibration whenever $\dist(t)>\dist(0)$. However, we can utilize the above freedom to construct a history that can support revival according to the BLP criterion. This is what we demonstrate next while preserving the raw record. To achieve this, we multiply the above history~\eqref{eq:depolarizing-trace} and scale~\eqref{eq:example-scale} by the same positive oscillatory factor,
\begin{equation}\label{eq:revival-reading}
 \begin{aligned}
 T_a(t)&=T_0e^{-\gamma t}\bigl(1-a+a\cos\omega t\bigr),\\
 \Lambda_a(t)&=\Lambda(t)\bigl(1-a+a\cos\omega t\bigr),
 \end{aligned}
\end{equation}
where $0<a<1/2$ and $\omega>0$. Notice that the common factor cancels in $2T_a/\Lambda_a=2T/\Lambda=\dist(t)$, so the spectral record is unchanged. The next task is to verify the bound inequality \eqref{eq:admissible} for this alternative history $T_a$, to which end we observe that $T_a(0)=T_0$, $0<1-2a\leq1-a+a\cos\omega t\leq1$ and $0<e^{-\gamma t}\leq1$ for $t\geq0$, so that
\[
 0<T_a(t)\leq T_a(0)\leq1.
\]

Having established this bound, we are now in a position to apply Proposition~\ref{prop:nonid} to $T_a$, and realize it with $f(t)=T_a(t)/T_0$ in the dephasing family~\eqref{eq:dephasing}. This bound compares $T_a(t)$ with its initial value, whereas a BLP revival requires an increase between later times. To see whether this admissible history also revives, we differentiate $T_a$ and evaluate at $t=3\pi/(2\omega)$,
\begin{equation}\label{eq:revival-derivative}
 \begin{aligned}
 \dot T_a(t)&=T_0e^{-\gamma t}\bigl[-\gamma(1-a+a\cos\omega t)-a\omega\sin\omega t\bigr],\\
 \dot T_a\!\left(\frac{3\pi}{2\omega}\right)
 &=T_0e^{-3\pi\gamma/(2\omega)}\bigl[a\omega-\gamma(1-a)\bigr].
 \end{aligned}
\end{equation}
This means a positive derivative at $t=3\pi/(2\omega)$ requires $a\omega>\gamma(1-a)$. Furthermore, requiring this time to lie inside a prescribed observation interval $[0,t_1]$, with $t_1>0$, additionally gives us $\omega>3\pi/(2t_1)$. Combining the two requirements, we thus obtain the following criterion,
\[
 \omega>\max\left\{\frac{\gamma(1-a)}{a},\frac{3\pi}{2t_1}\right\}.
\]
We have now established that the same spectral record~\eqref{eq:example-distance} can arise from the monotonically contracting depolarizing evolution or from a dephasing evolution with information backflow in the above examples, for any prescribed observation interval.

The remaining question is: what additional geometric data distinguish these two realizations, having the same qubit record but different Dirac scales? This is where our scalar-anchored geometry comes to the rescue as follows. Using the scalar radii $\ell_L=1/\mu_L$ and $\ell_R=1/\mu_R$ from~\eqref{eq:radii}, we obtain the following trace distance by plugging \eqref{eq:spectrum} into \eqref{eq:instant},
\begin{equation}\label{eq:reconstruction}
 T(t)=\frac{\dist(t)}{2}\sqrt{\ell_L(t)^{-2}+\ell_R(t)^{-2}}
 =\frac{\dist(t)\,\dist_{D(t)}(\omega_L,\omega_R)}{2\,\ell_L(t)\,\ell_R(t)},
\end{equation}
where the equality is due to \eqref{eq:pythagoras}. For our equal links, $\Lambda_a=h\Lambda$, with $h(t)=1-a+a\cos\omega t$, the corresponding radii become $\ell_{L,R}^{(a)}=\ell_{L,R}/h$. These radii therefore differ whenever $h\neq1$, and inserting them into~\eqref{eq:reconstruction} recovers $T$ and $T_a$, respectively. This means that, given the additional geometric data $\ell_{L,R}$, we can distinguish the two realizations, provided the radii belong to the respective Dirac geometry and are known at the same times as $\dist(t)$.

\begin{figure}[H]
 \centering
 \includegraphics[width=\linewidth]{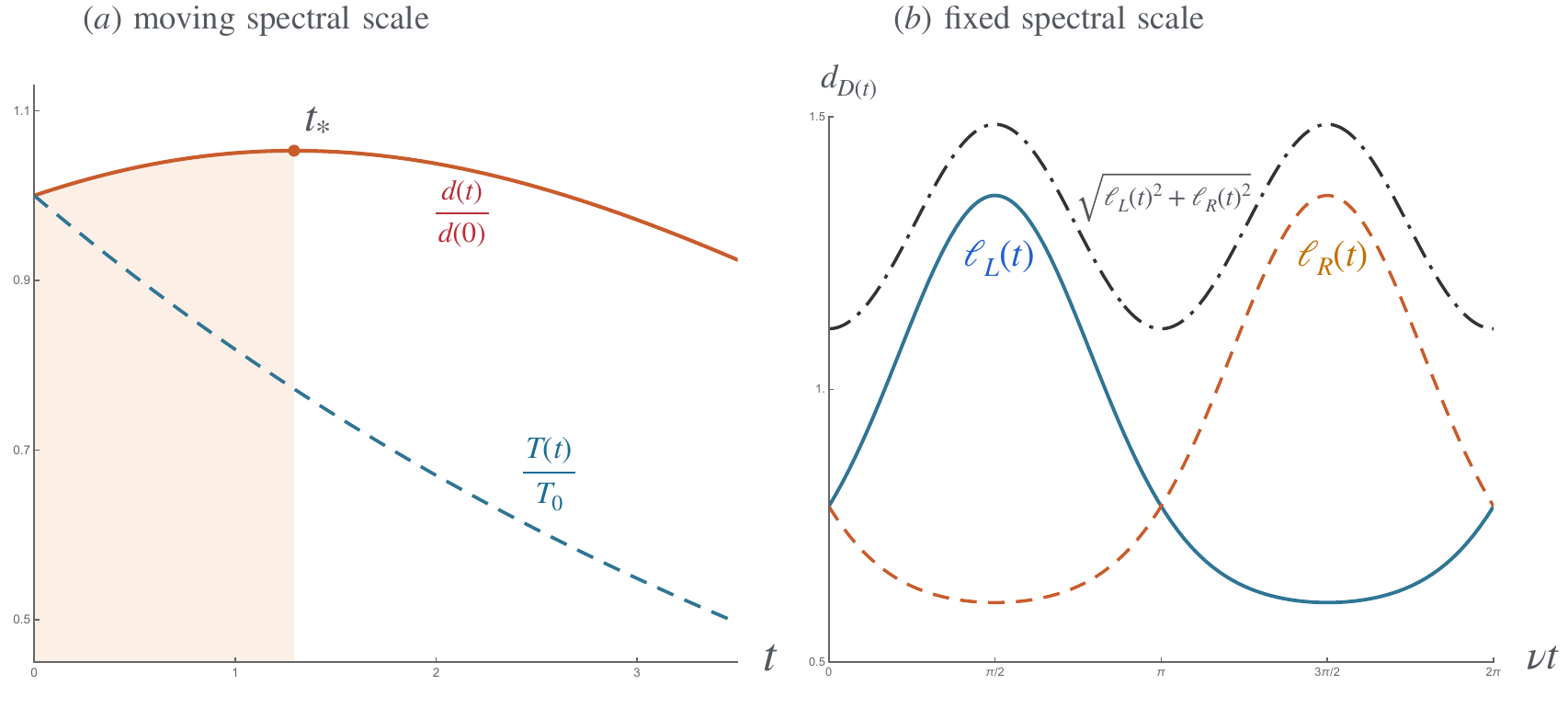}
 \caption{Illustration of the effects of moving Dirac links. (a) Normalized spectral and trace distances under depolarization, with the interval of spectral growth shown in the shaded region. (b) Scalar radii and scalar--scalar distance for an oscillating link imbalance at fixed $\Lambda$.}
 \label{fig:two-effects}
\end{figure}

\section{Scale and imbalance of the two links}\label{sec:links}

So far we have used the scalar radii to get the scale $\Lambda$ and thereby the trace distance. However, the qubit metric~\eqref{eq:static} depends on the two links only through this scale, and thus fails to determine their relative strengths. We now set out to examine this remaining freedom at fixed $\Lambda$, and introduce the scale-imbalance
\begin{equation}\label{eq:beta}
 \beta=\log\frac{\mu_R}{\mu_L}
      =\log\frac{\ell_L}{\ell_R}.
\end{equation}
We can now express the two link strengths in terms of the scale $\Lambda$ and this imbalance $\beta$ using \eqref{eq:spectrum} as follows,
\begin{equation}\label{eq:inverse}
 \mu_L=\frac{\Lambda}{\sqrt{1+e^{2\beta}}},
 \qquad
 \mu_R=\frac{\Lambda e^\beta}{\sqrt{1+e^{2\beta}}}.
\end{equation}
Thus $(\Lambda,\beta)\in(0,\infty)\times\mathbb R$ parametrizes all positive link pairs. We can now express the scalar radii and corresponding scalar--scalar spectral distance in terms of these two parameters by substituting \eqref{eq:inverse} into \eqref{eq:radii} and \eqref{eq:pythagoras},
\begin{equation}\label{eq:beta-distances}
 \ell_L=\frac{\sqrt{1+e^{2\beta}}}{\Lambda},\qquad
 \ell_R=\frac{\sqrt{1+e^{-2\beta}}}{\Lambda},\qquad
 \dist_D(\omega_L,\omega_R)=\frac{2\cosh\beta}{\Lambda}.
\end{equation}
Notice that the scalar--scalar distance is $\beta$-even, so the replacement $\beta\to-\beta$ simply exchanges the two links without changing this distance. Turning on time dependence, we obtain both rates from~\eqref{eq:spectrum} and~\eqref{eq:beta},
\begin{equation}\label{eq:link-rates}
 \frac{\dot\Lambda}{\Lambda}
 =\frac{\mu_L^2}{\Lambda^2}\frac{\dot\mu_L}{\mu_L}
  +\frac{\mu_R^2}{\Lambda^2}\frac{\dot\mu_R}{\mu_R},
 \qquad
 \dot\beta=\frac{\dot\mu_R}{\mu_R}-\frac{\dot\mu_L}{\mu_L}.
\end{equation}
Here the fractional link rates are coupled via the weights $\mu_{L,R}^2/\Lambda^2$ in the fractional scale rate, while their difference determines the change of $\beta$. So keeping $\Lambda$ fixed requires these weighted contributions to cancel. At an instant when $\mu_L=\mu_R$, this reduces to choosing equal and opposite fractional rates.

We therefore see that varying $\beta$ at fixed $\Lambda$ preserves the spectrum $\{-\Lambda,0,\Lambda\}$ as well as the qubit metric, but changes the scalar radii. The commutator seminorm~\eqref{eq:channels} still depends separately on $\mu_L$ and $\mu_R$ in~\eqref{eq:inverse}, so the scalar distances are sensitive to this change. To understand such an effect, let us consider $\Lambda=\Lambda_*$ and $\beta(t)=b\sin(\nu t)$ with $b,\nu>0$, so that the scalar--scalar distance \eqref{eq:beta-distances} becomes
\begin{equation}\label{eq:isospectral}
 \dist_{D(t)}(\omega_L,\omega_R)
 =\frac{2}{\Lambda_*}\cosh\bigl(b\sin(\nu t)\bigr)
\end{equation}
and oscillates between $2/\Lambda_*$ and $2\cosh b/\Lambda_*$. Moreover, the ratio $\ell_L/\ell_R=e^{b\sin(\nu t)}$ determines the sign of $\beta$, whereas the even function $\cosh$ in~\eqref{eq:isospectral} gives a scalar-distance period of $\pi/\nu$. Interestingly, over this half-cycle the radii exchange their values, so the labeled pair has the full period $2\pi/\nu$. Going further, we expand the scalar distance for small $b$ and find
\begin{equation}\label{eq:double-frequency}
 \dist_{D(t)}(\omega_L,\omega_R)
 =\frac{2}{\Lambda_*}
 \left[1+\frac{b^2}{4}\bigl(1-\cos(2\nu t)\bigr)+O(b^4)\right],
\end{equation}
so the leading oscillation is quadratic in $b$ and has twice the frequency of the imbalance. We illustrate these exact curves in Figure~\ref{fig:two-effects}(b) for $b=0.8$ and $\Lambda_*=1.8$.

\FloatBarrier
\section{Discussion and outlook}\label{sec:discussion}

We studied how time-dependent Dirac links affect the scale $\Lambda$ and the repercussions for Connes spectral and calibrated Helstrom geometry in the isotropic scalar-anchored qubit model. In particular, we addressed how the calibration changes induced by moving links combine with the separate changes in qubit distinguishability. We first considered a depolarizing evolution, for which a decreasing scale produces spectral growth in spite of the contracting trace distance. We then asked whether knowing the full spectral record could resolve this ambiguity without knowing the scale. Proposition~\ref{prop:nonid} shows that the same record can arise from both contracting and reviving trace-distance histories. The scalar references distinguish these histories through~\eqref{eq:reconstruction}, provided their distances are taken at the same times and within the same isotropic geometry. Having recovered the scale, we showed how Eq. \eqref{eq:rate} determines the sign of $\dot T$ and hence whether information backflow occurs.

On a related note, we studied the remaining freedom in the relative variation of link strengths as an imbalance $\beta$, once the scale $\Lambda$ has been fixed. This freedom leads to changing scalar distances, while the Dirac spectrum and qubit metric remain unaffected. On the geometry side, this amounts to an ambiguity when only the qubit metric is known. This is resolved by the additional geometric data of the two labeled scalar radii, whereas the scalar--scalar distance alone cannot distinguish an exchange of the links.

The reconstruction~\eqref{eq:reconstruction} assumes that the spectral distances are known. In practice, however, uncertainties associated with sampling of the scalar radii could propagate into the reconstructed trace distance. It would therefore be useful to track such a propagation and determine when an observed increase in the reconstructed trace distance can be distinguished from sampling errors. Contracting and reviving pairs are examples that would allow us to test such a criterion. It would also be interesting to explore whether the scalar-anchored model can be implemented in a physical qubit system, with controllable reference couplings and accessible calibration readings.

Another natural follow-up question concerns anisotropic modifications of the links. Allowing such direction-dependent Dirac links can deform the qubit metric, which could impart directional dependence to the spectral distance as well as a changing overall scale. In particular, we can no longer expect the identity~\eqref{eq:instant} to hold, and recovering the trace distance would require more than the scalar correction used here. It would also be useful to examine whether the Pythagorean spectral geometry of the doubled Moyal plane~\cite{DoubledMoyal} admits a comparable calibration of quantum-state distinguishability.

\subsection*{Acknowledgments}
The author acknowledges support from the DFG through a Walter Benjamin Fellowship, grant no.~515782239.


\begin{thebibliography}{10}
\bibitem{BLP}
Breuer, H.-P., Laine, E.-M., Piilo, J.: Measure for the degree of non-Markovian behavior of quantum processes in open systems. Phys. Rev. Lett. \textbf{103}, 210401 (2009). \href{https://doi.org/10.1103/PhysRevLett.103.210401}{doi:10.1103/PhysRevLett.103.210401}

\bibitem{Review}
Breuer, H.-P., Laine, E.-M., Piilo, J., Vacchini, B.: Colloquium: Non-Markovian dynamics in open quantum systems. Rev. Mod. Phys. \textbf{88}, 021002 (2016). \href{https://doi.org/10.1103/RevModPhys.88.021002}{doi:10.1103/RevModPhys.88.021002}

\bibitem{Connes}
Connes, A.: Noncommutative Geometry. Academic Press (1994)

\bibitem{DLV}
D'Andrea, F., Lizzi, F., V\'arilly, J.C.: Metric properties of the fuzzy sphere. Lett. Math. Phys. \textbf{103}, 183--205 (2013). \href{https://doi.org/10.1007/s11005-012-0590-5}{doi:10.1007/s11005-012-0590-5}

\bibitem{MoyalFuzzy}
Devi, Y.C., Kumar, K., Chakraborty, B., Scholtz, F.G.: Revisiting Connes' finite spectral distance on noncommutative spaces: Moyal plane and fuzzy sphere. Int. J. Geom. Methods Mod. Phys. \textbf{15}, 1850204 (2018). \href{https://doi.org/10.1142/S0219887818502043}{doi:10.1142/S0219887818502043}

\bibitem{GKS}
Gorini, V., Kossakowski, A., Sudarshan, E.C.G.: Completely positive dynamical semigroups of $N$-level systems. J. Math. Phys. \textbf{17}, 821--825 (1976). \href{https://doi.org/10.1063/1.522979}{doi:10.1063/1.522979}

\bibitem{Helstrom}
Helstrom, C.W.: Quantum Detection and Estimation Theory. Mathematics in Science and Engineering, vol.~123. Academic Press (1976)

\bibitem{IKM}
Iochum, B., Krajewski, T., Martinetti, P.: Distances in finite spaces from noncommutative geometry. J. Geom. Phys. \textbf{37}, 100--125 (2001). \href{https://doi.org/10.1016/S0393-0440(00)00044-9}{doi:10.1016/S0393-0440(00)00044-9}

\bibitem{Static}
Kumar, K.: Calibrated Helstrom geometry on the Bloch ball via Connes spectral distance. \href{https://arxiv.org/abs/2606.13824}{arXiv:2606.13824} (2026)

\bibitem{DoubledMoyal}
Kumar, K., Chakraborty, B.: Spectral distances on the doubled Moyal plane using Dirac eigenspinors. Phys. Rev. D \textbf{97}, 086019 (2018). \href{https://doi.org/10.1103/PhysRevD.97.086019}{doi:10.1103/PhysRevD.97.086019}

\end{thebibliography}
\end{document}